\documentclass[conference,letterpaper]{IEEEtran}

\usepackage[utf8]{inputenc} 
\usepackage[T1]{fontenc}
\usepackage{url}
\usepackage{ifthen}
\usepackage{cite}
\usepackage[cmex10]{amsmath}

\usepackage{graphicx}
\usepackage{dcolumn}
\usepackage{bm}
\usepackage{braket}
\usepackage{amsfonts}
\usepackage{multirow}
\usepackage{arydshln}
\usepackage{color}
\usepackage{amsthm,paralist}
\usepackage{mathtools}
\usepackage{amssymb}
\usepackage{graphicx}
\usepackage{tabularx}

\usepackage{amsfonts} 
\usepackage{mleftright}
\usepackage{blkarray}
\usepackage{enumitem}
\usepackage{nccmath}

\usepackage{times}
\usepackage{tikz}
\usepackage[hidelinks,colorlinks=true,urlcolor=blue,linkcolor=blue,citecolor=red]{hyperref}
\newcommand{\F}{\mathbb{F}}
\newtheorem{theorem}{Theorem}

\newtheorem{lemma}{Lemma}

\newcommand{\bl}[1]{{\color{blue}#1}}
\begin{document}
\title{Universal Communication Efficient Quantum Threshold Secret Sharing Schemes} 

\author{%
  \IEEEauthorblockN{Kaushik Senthoor and Pradeep Kiran Sarvepalli}
  \IEEEauthorblockA{Department of Electrical Engineering\\
  Indian Institute of Technology Madras\\
  Chennai 600 036, India
                    }
}

\maketitle

\begin{abstract}
 Quantum secret sharing (QSS) is a cryptographic protocol in which a quantum secret is distributed among a number of parties where some subsets of the parties are able to recover the secret while some subsets are unable to recover the secret. In the standard $((k,n))$ quantum threshold secret sharing scheme, any subset of $k$ or more parties out of the total $n$ parties can recover the secret while other subsets have no information about the secret. But recovery of the secret incurs a communication cost of at least $k$ qudits for every qudit in the secret. Recently, a class of communication efficient QSS schemes were proposed which can improve this communication cost to $\frac{d}{d-k+1}$ by contacting $d\geq k$ parties where $d$ is fixed prior to the distribution of shares. In this paper, we propose a more general class of $((k,n))$ quantum secret sharing schemes with low communication complexity. Our schemes are universal in the sense that the combiner can contact any number of parties to recover the secret with communication efficiency i.e. any $d$ in the range $k\leq d\leq n$ can be chosen by the combiner. This is the first such class of universal communication efficient quantum threshold schemes. 
\end{abstract}


\section{Introduction}

\noindent
{\em Motivation.}
A quantum secret sharing protocol enables the secure distribution of  a secret among mutually collaborating parties so that only certain collections of parties can recover the
secret. 
Since the proposal of quantum secret sharing for classical secrets by Hillery {\em et al.} \cite{hillery99} and its extension to share quantum secrets by Cleve {\em et al.} \cite{cleve99} there has been extensive research in this field \cite{gottesman00,karlsson99,smith99,imai03,markham08,senthoor19,ben12,ps10}. 
Quantum secret sharing schemes provide greater security than classical secret sharing schemes \cite{hillery99}.
Quantum secret sharing has been experimentally demonstrated by many groups \cite{tittel01,wei13,hao11,bogdanski08,bell14,schmid05,gaertner07,lance04}. 
In this paper we are interested in optimizing the resources needed for quantum secret sharing. 
Specifically, we propose communication efficient quantum threshold secret sharing schemes. 

The most popular quantum secret sharing scheme is the quantum threshold secret sharing scheme (QTS). 
In this scheme a minimum of $k$ players are required to recover the secret. 
It is often denoted as a $((k,n))$  scheme indicating that $k$ or more players out of the $n$
players can recover the secret. 
Such a scheme can share one secret qudit. 
The state given to each player is called the share of the player.
After the secret has been shared the players who plan to recover the secret combine their
shares together and reconstruct the secret. 
Alternatively, the parties involved in the recovery could communicate all or part of their share to a third party designated as the combiner.
The amount of quantum communication is called the communication complexity for recovery.
The standard method due to \cite{cleve99} requires the $m n $ qudits to be shared for share distribution and 
at least $mk$ qudits for recovery.

The analogous problem of reducing communication complexity has been studied classically \cite{wang08,bitar16,bitar18,huang16,huang17,penas18} but not as much in the quantum setting.
Only recently, Ref.~\cite{senthoor19} showed that the quantum communication cost during recovery can 
be reduced by using a subset of players whose cardinality is more than the threshold required to recover the secret.
The gains can be significant and for a $((k,2k-1))$ threshold scheme, they showed that the 
gains in communication complexity of recovery per secret qudit can be as large as $O(k)$. 
One limitation of those schemes was that these gains were only for a subset of players whose size $d$ was fixed.

\smallskip
\noindent
{\em Contribution.}
In this paper, we address the problem of designing quantum threshold schemes that are universal in that any subset of size $d\geq k$ would provide gains in communication cost during recovery. 
Our schemes generalize the classical schemes of \cite{bitar16,bitar18} to the quantum setting. 
We denote them as $((k,n,*))$ schemes.
In an earlier work \cite{senthoor19}, a construction for $((k,n,d))$ communication efficient QTS has been proposed. 
However, that construction only works for a fixed value of $d$ in the range of $k<d\leq n$. 
The value of $d$ is decided prior to encoding of the secret and cannot be changed. 
When $d$ parties are contacted,  the proposed construction achieves the same communication complexity as that of fixed $d$.
So there is no loss in communication complexity with the increased flexibility to change $d$.
This is the first such class of communication efficient quantum threshold secret sharing schemes where the number of parties contacted can be varied from 
$k$ to $n$.

\smallskip
\noindent
{\em Notation.}
We define the two qudit operator $L_\alpha$ as 
\begin{eqnarray}
L_\alpha\ket{i}_c\ket{j}_t  = \ket{i}_c\ket{j+\alpha i}_t, \label{eq:Lt-defn}
\end{eqnarray}
where  $i, j\in\F_q$ and $\alpha\in\F_q$ is a constant. The subscript $c$ and $t$ indicate that they are control and target qudits respectively. 
This operator generalizes the CNOT gate. 

We take the standard basis of $\mathbb{C}^q$ to be  $\{\ket{x}\mid x\in \F_q \}$.
We denote $\ket{x_1x_2\cdots x_\ell}$ by $\ket{\underline{x}}$ where $\underline{x}$ is the vector $(x_1,x_2,\hdots,x_\ell)$.
The standard basis for $\mathbb{C}^{q^\ell}$ is taken to be $\{\ket{\underline{x}}\mid \underline{x}\in \F_q^\ell \}$.
For any invertible matrix $K\in\F_q^{\ell\times\ell}$, we define the unitary operation $U_K$ 
\begin{eqnarray}
U_K\ket{\underline{x}}  = \ket{K\underline{x}} =\ket{\underline{y}},\label{eq:UK-defn} 
\end{eqnarray}
where $\underline{y}= (y_1,\ldots, y_n)$ and $y_i = \sum_{j}K_{ij}x_j$.
As the mapping $L_K:\underline{x}\mapsto K\underline{x}$ is a bijection from $\F_q^\ell$ to $\F_q^\ell$ for any invertible matrix $K$, clearly $U_K$ is a unitary operation.

Let $A=[a_{ij}]$ be an $m\times n$ matrix from $\F_q^{m\times n}$.
Then $\ket{A}$ indicates the state $\ket{a_{11}a_{21}\hdots a_{m1}}$$\ket{a_{12}a_{22}\hdots a_{m2}}$$\hdots$\ $\ket{a_{1n}a_{2n}\hdots a_{mn}}$.
Let $K$ be an invertible $m\times m$ matrix.
Then applying $K$ to the state $\ket{A}$ is defined as transforming state $\ket{A}$ to $\ket{KA}$ by $U_K^{\otimes n}$.

Consider the matrices $B_1,B_2,\hdots,B_f$ where each of these $f$ matrices has the same $n$ number of columns.
Then, we use the notation $\ket{A(B_1,B_2,\hdots,B_f)}$ to denote $\ket{A\left[\begin{array}{cccc} B_1^t & B_2^t & \hdots & B_{f}^t\end{array} \right]^t}$ and $\ket{B_1,B_2,\hdots,B_f}$ to denote $\ket{\left[\begin{array}{cccc} B_1^t & B_2^t & \hdots & B_{f}^t\end{array} \right]^t}$.

We  use the notation $[n]:=\{1,2,\ldots, n \}$ and $[i,j]:=\{i, i+1,\ldots, j \}$.
Let $V$ be a $m\times n$ matrix  and  $A \subseteq [m]$, $B\subseteq [n]$.
We denote by $V_A$, the submatrix of 
$V$ formed by taking the rows indexed by entries in $A$.
Similarly, we can form a submatrix of $V$ by taking the columns of $V$. This is indicated as 
$V^B$. 
We can also form a submatrix $V_A^B$ which takes some rows and columns from $V$.
 
\smallskip
\noindent
{\em Illustration.}
In this section, we give an example to illustrate the gains in communication complexity for a suitably designed quantum threshold scheme. 
Later sections in this paper provide a construction for such universal communication efficient quantum secret sharing schemes.
A running example for the proposed construction is included in the paper.

Consider a secret of three qudits with each qudit of dimension 11. This secret will be encoded into 15 qudits, giving three qudits to each of the five parties. 
Every qudit is of dimension 11. 
Define matrices $V$ and $M$ as follows. 
\begin{equation*}
V=
\begin{bmatrix}
9&3&4&6&1\\2&9&3&4&6\\8&2&9&3&4\\7&8&2&9&3\\5&7&8&2&9
\end{bmatrix}
\text{and\ }
M=
\left[
\begin{tabular}{ccc}
$s_1$&0&0\\$s_2$&$r_1$&0\\$s_3$&$r_2$&$r_3$\\$r_1$&$r_3$&$r_5$\\$r_2$&$r_4$&$r_6$
\end{tabular}
\right].
\end{equation*}

Here $V$ is a Cauchy matrix.
Then the encoding for a universal communication efficient QTS scheme is given by the following mapping
\begin{eqnarray}
\ket{s_1 s_2 s_3}\mapsto\sum_{\underline{r}\in\F_{11}^6}\ket{c_{11}c_{12}c_{13}}&&\!\!\ket{c_{21}c_{22}c_{23}}\ket{c_{31}c_{32}c_{33}}
\\[-0.5cm]&&\ \ \ \ \ket{c_{41}c_{42}c_{43}}\ket{c_{51}c_{52}c_{53}}\nonumber
\label{eq:enc_qudits_3_5_4}
\end{eqnarray}
where $\underline{r}=(r_1,r_2,\ldots, r_6)$
and 
$c_{ij} $ is the $(i,j)$th entry of $C=VM$. 

When combiner requests $d=5$ parties, they send the first qudit their shares, namely $c_{i,1}$. 
When $d=4$, the combiner accesses the first two qudits of each share of the four parties contacted. 
When $d=3$, the combiner accesses all three qudits of the share of the three parties contacted.

Consider the case when $d=5$ i.e. the first qudits from all five parties are accessed. Applying the operation $U_{V^{-1}}$ on these five qudits, we obtain
\begin{eqnarray}
\ket{s_1 s_2 s_3}\sum_{{\underline{r}\in\F_{11}^6}}\ket{r_1 r_2}\ket{c_{12}c_{22}c_{32}c_{42}c_{52}}\ket{c_{13}c_{23}c_{33}c_{43}c_{53}}
\end{eqnarray}

Consider the case when $d=4$. Assume that the first four parties are accessed. The first two qudits from the four parties are accessed. Applying the operation $U_{K_1}$ on the set of four second qudits, where $K_1$ is the inverse of $V_{[4]}^{[2,5]}$
we obtain
\begin{eqnarray}
\sum_{{\underline{r}\in\F_{11}^6}}\ket{c_{11}c_{21}c_{31}c_{41}c_{51}}\ket{r_1 r_2 r_3 r_4}\ket{c_{52}}\ket{c_{13}c_{23}c_{33}c_{43}c_{53}}.\label{eq:d4-recovery-1}
\end{eqnarray}

Then, on applying the operators $L_{10}\ket{r_2}\ket{c_{11}}$, $L_5\ket{r_2}\ket{c_{21}}$, $L_7\ket{r_2}\ket{c_{31}}$ and $L_8\ket{r_2}\ket{c_{41}}$, see Eq.~\eqref{eq:d4-recovery-1}, we obtain
\begin{eqnarray*}
\sum_{{\underline{r}\in\F_{11}^6}}&&\ket{9s_1+3s_2+4s_3+6r_1}\ket{2s_1+9s_2+3s_3+4r_1}\nonumber
\\[-0.4cm]&&\,\,\,\ket{8s_1+2s_2+9s_3+3r_1}\ket{7s_1+8s_2+2s_3+9r_1}\ket{c_{51}}\nonumber
\\&&\ \ \ \,\ket{r_1 r_2 r_3 r_4} \ket{c_{52}}\ket{c_{13}c_{23}c_{33}c_{43}c_{53}}.
\end{eqnarray*}

Applying the operation $U_{K_4}$ on the set of four first qudits, where $K_4$ is the inverse of $V_{[4]}^{[4]}$,
we obtain
\begin{eqnarray*}
\ket{s_1 s_2 s_3}\sum_{{\underline{r}\in\F_{11}^6}}\ket{r_1}\ket{c_{51}}\ket{r_1 r_2 r_3 r_4}\ket{c_{52}}\ket{c_{13}c_{23}c_{33}c_{43}c_{53}}.
\end{eqnarray*}

Then, on applying suitable $L_\alpha$ operators, we obtain
\begin{eqnarray*}
&&\ket{s_1 s_2 s_3}\sum_{{\underline{r}\in\F_{11}^6}}\ket{r_1}\ket{c_{51}}\ket{r_1 c_{51} r_3 c_{52}}\ket{c_{52}}\ket{c_{13}c_{23}c_{33}c_{43}c_{53}}\nonumber
\\&&=\ket{s_1 s_2 s_3}\sum_{\underline{r}'\in\F_{11}^6}\ket{r_1}\ket{r_2'}\ket{r_1 r_2' r_3 r_4'}\ket{r_4'}\ket{c_{13}c_{23}c_{33}c_{43}c_{53}}.
\end{eqnarray*}
where $\underline{r}'=(r_1,r_2',r_3,r_4',r_5,r_6)$.

In contrast, for the standard $((3,5))$ QTS due to Cleve {\em et al.} 3 qudits need to be communicated for recovery of 1 qudit of secret. In the $((3,5,5))$ fixed $d$ communication efficient QTS scheme from \cite{senthoor19}, 5 qudits need to be communicated for recovery of 3 qudits i.e 5/3 qudits per 1 qudit of secret. But this scheme does not provide the flexibility of contacting four parties communication efficiently. The scheme provided above can solve that problem. It provides communication efficiency at $d=5$ and as well as $d=4$. However, at $d=4$, this scheme gives communication cost of 8 qudits to recover secret of 3 qudits i.e. 8/3 qudits per one qudit of secret whereas the $((3,5,4))$ fixed $d$ communication efficient QTS gives 2 qudits per one qudit of secret. Our proposed construction below can provide the same communication efficiency as the fixed $d$ communication efficient QTS schemes at both $d=4$ and $d=5$.

\section{Background}
A quantum secret sharing (QSS) scheme is a protocol to encode the secret in arbitrary quantum state and share it among $n$ parties such that certain subsets of parties, called authorized sets, can recover the secret (recoverability) and certain subsets of parties, called unauthorized sets, do not have any information on the secret (secrecy).
A QSS scheme is called perfect quantum secret sharing scheme if any subset of the $n$ parties is either an authorized set or an unauthorized set. 
We focus on the $((k,n))$ quantum threshold schemes (QTS), where there are $n$ players and any $k$ or more players can recover the secret while fewer than $k$ players have no information about the secret. 

The realization of a quantum secret sharing is specified by giving an encoding for the basis states of the secret. 
Any encoding has to satisfy the properties of recoverablity and secrecy to realize a valid QSS. 
The recoverability constraint implies that any authorized set must be able to recover the secret and the secrecy constraint implies that sets that are unauthorized cannot recover the secret. 
In this paper, in Section~\ref{ss:iii_a}, we describe the encoding for the proposed construction of $((k,n,*))$ universal communication efficient quantum threshold secret sharing scheme. 
In Section~\ref{ss:iii_b}, the proof for secret recovery is given. In Section~\ref{ss:iii_c}, we show that our construction satisfies the secrecy constraint.

\section{Universal communication efficient QTS}\label{s:iii}
\subsection{Encoding}\label{ss:iii_a}
\noindent

Communication efficient quantum secret sharing schemes for particular values of $k$ and $n=2k-1$ can be designed to work for all possible values of $d$ in the range $k$ through 
$n$ where $k \leq d \leq n$.
We introduce the following terms before discussing the scheme. For $1\leq i\leq k$,
\begin{subequations}\label{eq:ceqts-params}
\begin{eqnarray}
d_i = n&-&i+1=2k-i\\
m_i = d_i-k+1, &m& = \text{lcm}\{m_1,m_2,\hdots m_k\}\\
a_i = m/(\!&d_i&-k+1)\\
b_1=a_1,\ b_i = &a_{i}& -\ a_{i-1} \text{ for }i>1  
\end{eqnarray}
\end{subequations}
Here $m$ is the total number of secret qudits shared. 
The total number of qudits with each party is also given by $m$.
This is consistent with the fact that in a perfect secret sharing scheme the size of the share must be at least as large as the secret \cite{gottesman00,imai03}.

Now $a_i$ gives the number of qudits communicated from each accessible share when $d_i$ shares are accessed to recover the secret. 
This means that $a_id_i$ qudits are communicated to the combiner when $d_i$ players are contacted. 
Let $b_1=a_1$ and $b_i=a_i-a_{i-1}$ for $2\leq i\leq k$. Pick a prime number $q\geq 2(2k-1)$.
Consider the secret $s=(s_1, s_2, \ldots, s_m)\in\mathbb{F}_q^m$ and $\underline{r}=(r_1, r_2, \ldots, r_{m(k-1)})\in \mathbb{F}_q^{m(k-1)}$. 

Entries in $\underline{s}$ are rearranged into the matrix $S$ of size $k\times (m/k)$.\vspace{-0.25cm}
\begin{eqnarray}
S= \left[\begin{array}{cccc} 
s_1&s_{k+1}& \cdots & s_{m-k+1}\\
s_2&s_{k+2}& \cdots & s_{m-k+2}\\
\vdots&\vdots& \ddots & \vdots\\
s_k&s_{2k}& \cdots & s_{m}\\
\end{array} \right]\label{eq:secret-ceqss}
\end{eqnarray}

Entries in $\underline{r}$ are rearranged into $k$ matrices i.e. $R_1$ of size $(k-1) \times b_1$, $R_2$  of size $(k-1)\times b_2$ and so on till $R_k$ of size $(k-1)\times b_k$.
\begin{eqnarray}
R_1=
\left[\begin{array}{cccc}
r_1& r_k& \cdots & r_{(a_1-1)(k-1)+1} \\
r_2& r_{k+1}& \cdots & r_{(a_1-1)(k-1)+2}\\
\vdots& \vdots & \ddots & \vdots\\
r_{k-1}& r_{2(k-1)}& \cdots & r_{a_1(k-1)}
\end{array} \right]\nonumber
\end{eqnarray}
For $2\leq i\leq k$, $R_i$ is given by 
\begin{eqnarray}
\!\left[\!\!\begin{array}{cccc}
r_{a_{i-1}(k-1)+1}& r_{(a_{i-1}+1)(k-1)+1}& \cdots & r_{(a_i-1)(k-1)+1} \\
r_{a_{i-1}(k-1)+2}& r_{(a_{i-1}+1)(k-1)+2}& \cdots & r_{(a_i-1)(k-1)+2}\\
\vdots& \vdots & \ddots & \vdots\\
r_{(a_{i-1}+1)(k-1)}& r_{(a_{i-1}+2)(k-1)}& \cdots & r_{a_i(k-1)}
\end{array}\!\! \right].\nonumber
\end{eqnarray}
The  matrix $C$, called code matrix,  is defined as follows.
\begin{eqnarray*}
C  = V   M 
\end{eqnarray*}
where
\begin{eqnarray*}
M =
\left[
\begin{tabular}{c:c:c:c:c}
\multirow{4}{*}{$\ S\ $} & {\large \ 0\ } & \multirow{2}{*}{\large 0} & \multirow{4}{*}{$\ \ddots\ $} & \multirow{3}{*}{\large 0}\\ \cdashline{2-2}
&\multirow{3}{*}{$D_1$} & &\\ \cdashline{3-3}
& & \multirow{2}{*}{$D_2$} &\\ \cdashline{5-5}
& & & & $\ \ D_{k-1}\ \ $\\
\cdashline{1-5}
\multirow{2}{*}{$R_1$} & \multirow{2}{*}{$R_2$} & \multirow{2}{*}{$R_3$} & \multirow{2}{*}{$\hdots$} & \multirow{2}{*}{$R_k$}\\
& & & &\\
\end{tabular}
\right]
\end{eqnarray*}
and $V$ is an $n\times n$ Cauchy matrix over $\mathbb{F}_q$. 
Here, $D_i$ of size $(k-i)\times b_{i+1}$ is constructed by rearranging the entries in $i$th row of the matrix $[R_1\ R_2\hdots\ R_i]$.

The encoding for a universal communication efficient QTS is given as follows:
\begin{eqnarray}
\ket{s_1 s_2\hdots s_m}\ \mapsto\sum_{\underline{r}\in\mathbb{F}_q^{m(k-1)}}
\ \bigotimes_{u=1}^{n}\ \ket{c_{u,1} c_{u,2}\hdots c_{u,m}} \label{eq:enc_qudits_univ_d}
\end{eqnarray}
where $c_{ij}$  is the entry in $C$ from $i$th row and $j$th column.
For example, take $k=3$. This gives
\begin{eqnarray*}
n=2k-1=5, q=11\ \ \nonumber\\
d_1=5, d_2=4, d_3=3\ \ \nonumber\\
m_1=3, m_2=2, m_3=1\nonumber\\
m=\text{lcm}\{3,2,1\}=6\ \ \ \nonumber\\
a_1=2, a_2=3,a_3=6\ \ \nonumber\\
b_1=2,b_2=1, b_3=3\ \ \nonumber
\end{eqnarray*}
Let $q=11$. 
Then $C$, the coding matrix for $k=3$ is given as
\begin{eqnarray*}
\left[
\begin{tabular}{ccccc}
9&3&4&6&1\\2&9&3&4&6\\8&2&9&3&4\\7&8&2&9&3\\5&7&8&2&9
\end{tabular}
\right]
\left[
\begin{tabular}{cc:c:ccc}
$s_1$ & $s_4$ & 0 & 0 & 0 & 0\\
$s_2$ & $s_5$ & $r_1$ & 0 & 0 & 0\\
$s_3$ & $s_6$ & $r_3$ & $r_2$ & $r_4$ & $r_6$\\\hdashline
$r_1$ & $r_3$ & $r_5$ & $r_7$ & $r_9$ & $r_{11}$\\
$r_2$ & $r_4$ & $r_6$ & $r_8$ & $r_{10}$ & $r_{12}$
\end{tabular}
\right].
\end{eqnarray*}

Each entry in matrix $C$, $c_{ij}$ is a function of $\underline{s}$ and $\underline{r}$. 
However, note that the $D_i$ are functions of $\underline{r}$ alone. 

The encoding for the $((3,5,*))$ schemes is given by Eq.~\eqref{eq:enc_qudits_univ_d}. For example, the corresponding
$c_{ij}$ of the third share are given below. 
\begin{eqnarray*}
c_{31}&=&8s_1+2s_2+9s_3+3r_1+4r_2,\nonumber
\\c_{32}&=&8s_4+2s_5+9s_6+3r_3+4r_4,\nonumber
\\c_{33}&=&2r_1+9r_3+3r_5+4r_6,\nonumber
\\c_{34}&=&9r_2+3r_7+4r_8,\nonumber
\\c_{35}&=&9r_4+3r_9+4r_{10},\nonumber
\\c_{36}&=&9r_6+3r_{11}+4r_{12}.
\end{eqnarray*}

Our encoding matrix is somewhat similar to the matrix used in \cite{bitar16,bitar18}. 
However, there are some minor structural differences. 
Since we encoding quantum states in superposition, there is no need for 
generating random bits. 
Furthermore, due to the No-Cloning theorem, the total number of parties cannot exceed 
$2k-1$.

\subsection{Reconstruction of the secret}\label{ss:iii_b}
The combiner can reconstruct the secret depending upon the choice of $d$. 
Once $d=d_i$ is chosen, the combiner contacts a set of any $d_i$ parties to reconstruct the secret. 
Each of the contacted party sends $a_i=\frac{m}{d_i-k+1}$ qudits to the combiner. 
In total, the combiner has $\frac{d_im}{d_i-k+1}=a_id_i$ qudits.

With respect to the $((3,5,*))$ example in the previous section, suppose that the third party is contacted for reconstruction. 
If the party  belongs to recovery set of size $d_1=5$, then $a_1=2$ qudits are communicated to the combiner. 
Similarly, if  $d_2=4$, then $a_2=3$  and if $d_3=3$, then all the $a_3=6$ qudits are sent. 

The secret reconstruction happens in two stages. First, the basis states of the secret are reconstructed through suitable unitary operations. 
The classical secret sharing schemes stop the reconstruction at this point. 
But, the qudits containing the basis states of the secret can be entangled with the remaining qudits. 
So, in the second stage, the secret is extracted into a set of  qudits that are disentangled with the remaining qudits. 
\begin{lemma}[Secret recovery]\label{lm:recovery}
For a $((k,2k-1,*))$ scheme with the encoding  
given in Eq.~\eqref{eq:enc_qudits_univ_d}, we can recover the secret from any 
$d=2k-i$ shares where $1\leq i\leq k$ by accessing only the first $a_i=\frac{m}{d-k+1}$ qudits from each share where 
$m$ is as in Eq.~\eqref{eq:ceqts-params}.
\end{lemma}
\begin{proof} 
Each of the $d$ participants sends their first $a_i$ qudits to the combiner for reconstructing the secret. 
Let $D = \{j_1, j_2, \hdots, j_d\} \subseteq \{1,2,\hdots,2k-1\}$ be the set of $d$ shares chosen and $E=\{j_{d+1},j_{d+2},\hdots,j_{2k-1}\}$ be the complement of $D$. 
Then, Eq.~\eqref{eq:enc_qudits_univ_d} can be rearranged as
\begin{eqnarray}
\sum_{\underline{r}\in\mathbb{F}_q^{m(k-1)}}
&&\textcolor{blue}{\ket{c_{j_1,1}c_{j_2,1}...c_{j_d,1}}
\ket{c_{j_1,2}c_{j_2,2}...c_{j_d,2}}}\nonumber
\\[-0.2in]&&\textcolor{blue}{\ \ \ \ \ \ \hdots\ket{c_{j_1,a}c_{j_2,a}...c_{j_d,a}}}\nonumber
\\&&\ \ \ket{c_{j_{d+1},1}c_{j_{d+2},1}...c_{j_n,1}}
\ket{c_{j_{d+1},2}c_{j_{d+2},2}...c_{j_n,2}}\nonumber
\\&&\ \ \ \ \ \ \ \ \hdots\ket{c_{j_{d+1},a}c_{j_{d+2},a}...c_{j_n,a}}\nonumber
\\&&\ \ \ \ \ket{c_{1,a+1}c_{2,a+1}...c_{n,a+1}}
\ket{c_{1,a+2}c_{2,a+2}...c_{n,a+2}}\nonumber
\\&&\ \ \ \ \ \ \ \ \ \ \ \ \hdots\ket{c_{1,m}c_{2,m}...c_{n,m}}
\label{eq:acc_qudits}
\end{eqnarray}
where we have highlighted (in blue) the basis states of the qudits communicated to the combiner.
For the sake of exposition we will first cover the case of $i=1$ i.e. $d_i=2k-1$ where all the parties are contacted for their first 
$a_1$ qudits by the combiner. 
\\\\\textit{Case (i): $i=1$}
\\For $i=1$, $d=2k-1=n$. Now Eq.~\eqref{eq:acc_qudits} can be rewritten as
\begin{eqnarray*}
\sum_{\underline{r}\in\mathbb{F}_q^{m(k-1)}}
\textcolor{blue}{\ket{V(S,R_1)}}&&\ket{V(0,D_1,R_2)}\ket{V(0,D_2,R_3)}\nonumber
\\[-0.5cm]&&\ \ \ \ \ \ \ \ \ \ \ \ \ \ \ \ \ \ \ \ \hdots\ket{V(0,D_{k-1},R_k)}
\end{eqnarray*}

Since $V$ is an $n\times n$ Cauchy matrix and therefore invertible, we can apply ${V}^{-1}$ to the state 
$\ket{V(S, R_1)}$ and rearrange the qudits to obtain 
\begin{eqnarray*}
\textcolor{blue}{\ket{S}}\sum_{\underline{r}\in\mathbb{F}_q^{m(k-1)}}
\textcolor{blue}{\ket{R_1}}&&\ket{V(0,D_1,R_2)}\ket{V(0,D_2,R_3)}\nonumber
\\[-0.5cm]&&\ \ \ \ \ \ \ \ \ \ \ \ \ \ \ \ \ \ \ \ \ \ \ \hdots\ket{V(0,D_{k-1},R_k)}.
\end{eqnarray*}
We can clearly see that the secret is unentangled with the rest of the qudits.
Therefore, we can recover arbitrary superpositions also. 

\noindent
\\\textit{Case (ii): $2\leq i\leq k$: } Under this case, the state of the system is as follows. (This is the same as
Eq.~\eqref{eq:acc_qudits}, only the qudits in possession of the combiner have been rearranged and highlighted.)
\begin{eqnarray*}
\sum_{\substack{\underline{r}\in\\\mathbb{F}_q^{m(k-1)}}}
&&\!\!\!\textcolor{blue}{\ket{V_D(S,R_1)}\ \ket{V_D(0,D_1,R_2)}\hdots\ket{V_D(0,D_{i-1},R_i)}}\nonumber
\\[-0.7cm]&&\ket{V_E(S,R_1)}\ \ket{V_E(0,D_1,R_2)}\hdots\ket{V_E(0,D_{i-1},R_i)}\nonumber
\\&&\ \ \ket{V(0,D_i,R_{i+1})}\hdots\ket{V(0,D_{k-1},R_k)}\nonumber
\end{eqnarray*}
\vspace{-0.6cm}
\begin{eqnarray*}
=\!\!\sum_{\substack{\underline{r}\in\\\mathbb{F}_q^{m(k-1)}}}
&&\textcolor{blue}{\!\!\!\!\ket{V_D(S,R_1)}\ket{{V_D}^{[2,n]}(D_1,R_2)}\hdots\ket{{V_D}^{[i,n]}(D_{i-1},R_i)}}\nonumber
\\[-0.7cm]&&\ket{V_E(S,R_1)}\ket{V_E(0,D_1,R_2)}\hdots\ket{V_E(0,D_{i-1},R_i)}\nonumber
\\&&\ \ \ \ \ \ \ \ \ket{V(0,D_i,R_{i+1})}\hdots\ket{V(0,D_{k-1},R_k)}\nonumber
\end{eqnarray*}
Since ${V_D}^{[i,n]}$ is a $d_i\times d_i$ Cauchy matrix and therefore invertible, the combiner can apply the inverse of ${V_D}^{[i,n]}$ to 
$\ket{V_D^{[i,n]}(D_{i-1}, R_{i})}$ to transform the state as follows. 
\begin{eqnarray*}
\sum_{\underline{r}\in\mathbb{F}_q^{m(k-1)}}\hspace{-0.5cm}
&&\textcolor{blue}\ {\ket{V_D(S,R_1)}\ \ket{{V_D}^{[2,n]}(D_1,R_2)}}\nonumber
\\[-0.5cm]&&\ \ \ \ \ \ \ \ \textcolor{blue}{\hdots\ket{{V_D}^{[i-1,n]}(D_{i-2},R_{i-1})}\ \ket{D_{i-1}}\ket{R_i}}\nonumber
\\&&\ \ \ \ \ket{V_E(S,R_1)}\ \ket{V_E(0,D_1,R_2)}\hdots\ket{V_E(0,D_{i-1},R_i)}\nonumber
\\&&\ \ \ \ \ \ \ \ \ket{V(0,D_i,R_{i+1})}\hdots\ket{V(0,D_{k-1},R_k)}\nonumber
\end{eqnarray*}
Note that the matrix $D_{i-1}$ contains elements from the $(i-1)$th row of $R_{i-1}$.
Rearranging the qudits, we get
\begin{eqnarray*}
\sum_{\substack{\underline{r}\in\\\mathbb{F}_q^{m(k-1)}}}
&&\!\!\!\!\textcolor{blue}{\ket{V_D(S,R_1)}\ket{{V_D}^{[2,n]}(D_1,R_2)}\hdots\ket{{V_D}^{[i-2,n]}(D_{i-3},R_{i-2})}}\nonumber
\\[-0.7cm]&&\ \ \textcolor{blue}{\ket{W_{i-1}(D_{i-2},R_{i-1})}\ \ket{D_{i-1}\backslash R_{i-1}}\ket{R_i}}\nonumber
\\&&\ \ \ \ \ket{V_E(S,R_1)}\ \ket{V_E(0,D_1,R_2)}\hdots\ket{V_E(0,D_{i-1},R_i)}\nonumber
\\&&\ \ \ \ \ \ \ket{V(0,D_i,R_{i+1})}\hdots\ket{V(0,D_{k-1},R_k)}\nonumber
\end{eqnarray*}
where
$W_\ell = [{{V_D}^{[\ell,n]}}^t\ \underline{w}_{\ell,k+1}\ \underline{w}_{\ell,k+2}\hdots\underline{w}_{\ell,k+i-\ell}]^t$ for $1\leq\ell\leq i-1$ where $\underline{w}_{\ell,j}$ is a column vector of length $(2k-\ell)$ with one in the $j$th position and zeros elsewhere. $W_\ell$ is a $(2k-\ell)\times(2k-\ell)$ full-rank matrix and invertible. 
We have split the state $\ket{D_{i-1}}$ as $\ket{D_{i-1}\setminus R_{i-1}} \ket{ D_{i-1}\cap R_{i-1}}$. Then we merge $\ket{D_{i-1}\cap R_{i-1}}$ with 
$\ket{V_D^{[i-1,n]}(D_{i-2},R_{i-1})}$ to give $\ket{W_{i-1}(D_{i-2},R_{i-1})}$.

Now applying $W_{i-1}^{-1}$ to the state $\ket{W_{i-1}(D_{i-2},R_{i-1})}$, we are able to extract $D_{i-2}$ and $R_{i-1}$ as shown below:
\begin{eqnarray*}
\sum_{\underline{r}\in\mathbb{F}_q^{m(k-1)}}\hspace{-0.5cm}
&&\textcolor{blue}{\ket{V_D(S,R_1)}\ket{{V_D}^{[2,n]}(D_1,R_2)}}\nonumber
\\[-0.5cm]&&\ \ \ \ \ \ \ \ \ \ \ \ \ \ \ \ \ \ \ \ \ \ \ \ \ \textcolor{blue}{\hdots\ket{{V_D}^{[i-2,n]}(D_{i-3},R_{i-2})}}\nonumber
\\&&\ \ \ \ \textcolor{blue}{\ket{D_{i-2}}\ket{R_{i-1}}\ \ket{D_{i-1}\backslash R_{i-1}}\ket{R_i}}\nonumber
\\&&\ \ \ \ \  \ket{V_E(S,R_1)}\ \ket{V_E(0,D_1,R_2)}\hdots\ket{V_E(0,D_{i-1},R_i)}\nonumber
\\&&\ \ \ \ \   \ket{V(0,D_i,R_{i+1})}\hdots\ket{V(0,D_{k-1},R_k)}\nonumber
\end{eqnarray*}
Now we repeat the process with $D_{i-2}$ and $R_{i-2}$ to extract $D_{i-3}$ and $R_{i-2}$. 
Rearranging the qudits, we obtain,
\begin{eqnarray*}
\sum_{\underline{r}\in\mathbb{F}_q^{m(k-1)}}\hspace{-0.5cm}
&&\textcolor{blue}{\ket{V_D(S,R_1)}\ket{{V_D}^{[2,n]}(D_1,R_2)}}\nonumber
\\[-0.5cm]&&\ \ \ \ \ \ \ \ \ \ \ \ \ \ \ \ \ \ \ \ \ \ \ \ \ \textcolor{blue}{\hdots\ket{W_{i-2}(D_{i-3},R_{i-2})}}\nonumber
\\&&\ \ \textcolor{blue}{\ket{D_{i-2}\backslash R_{i-2}}\ket{R_{i-1}}\ \ket{D_{i-1}\backslash \{R_{i-1},R_{i-2}\}}\ket{R_i}}\nonumber
\\&&\ \ \ \ket{V_E(S,R_1)}\ \ket{V_E(0,D_1,R_2)}\hdots\ket{V_E(0,D_{i-1},R_i)}\nonumber
\\&&\ \ \ \ \ket{V(0,D_i,R_{i+1})}\hdots\ket{V(0,D_{k-1},R_k)}\nonumber
\end{eqnarray*}
Repeating this process for all $D_{i-3}, R_{i-2}$ through $D_1, R_2$ and $S, R_1$, and 
applying the inverses of $W_{i-3}, W_{i-4},\hdots W_{1}$ in successive steps to the suitable sets of qudits and rearranging, we obtain,

\begin{eqnarray*}
\textcolor{blue}{\ket{S}}\!\!\sum_{\substack{\underline{r}\in\\\mathbb{F}_q^{m(k-1)}}}\hspace{-0.5cm}
&&\textcolor{blue}\ \ \bl{\ket{R_1}\ket{R_2}\ket{R_3}\hdots \ket{R_i}}\nonumber
\\[-0.7cm]&&\ \ \ \ket{V_E(S,R_1)}\ket{V_E(0,D_1,R_2)}\hdots\ket{V_E(0,D_{i-1},R_i)}\nonumber
\\&&\ \ \ \ \ket{V(0,D_i,R_{i+1})}\hdots\ket{V(0,D_{k-1},R_k)}\nonumber
\end{eqnarray*}
Let $J_\ell=[\ell-1]\cup[\ell+1,k-1]$ for $1\leq \ell\leq i-1$. Since $D_{i-1}$ is formed from the $(i-1)$th rows of $R_1, R_2, \hdots, R_{i-1}$, the qudits can be rearranged to obtain,
\begin{eqnarray*}
\textcolor{blue}{\ket{S}}\!\!\sum_{\substack{\underline{r}\in\\\mathbb{F}_q^{m(k-1)}}}
&&\bl{\ket{R_{1,J_{i-1}}}\ket{R_{2,J_{i-1}}}\hdots \ket{R_{{i-1},J_{i-1}}}\ket{D_{i-1}}\ket{R_{i}}}\nonumber
\\[-0.7cm]&&\ \ket{V_E(S,R_1)}\ket{V_E(0,D_1,R_2)}\hdots\ket{V_E(0,D_{i-1},R_i)}\nonumber
\\&&\ \ \ket{V(0,D_i,R_{i+1})}\hdots\ket{V(0,D_{k-1},R_k)}\nonumber
\end{eqnarray*}
Consider the matrix
\renewcommand{\arraystretch}{1.7}
\begin{equation}
G_\ell=
\left[
\begin{tabular}{ccc}
$I_{k-i+\ell}$ & \multicolumn{2}{c}{0}\\\hdashline
\multicolumn{3}{c}{$V_E^{[i-\ell+1,n]}$}\\\hdashline
\multicolumn{2}{c}{0} & $I_{k-i}$
\end{tabular}
\right].
\end{equation}
$G_\ell$ is a $(d_i+\ell-1)\times(d_i+\ell-1)$ invertible matrix. Applying $G_1$ on $\ket{D_{i-1}}\ket{R_{i}}$, we obtain,
\begin{eqnarray*}
\textcolor{blue}{\ket{S}}\!\!\sum_{\substack{\underline{r}\in\\\mathbb{F}_q^{m(k-1)}}}
\!\!\!\!&&\bl{\ket{R_{1,J_{i-1}}}\ket{R_{2,J_{i-1}}}\hdots \ket{R_{{i-1},J_{i-1}}}}\nonumber
\\[-0.7cm]&&\ \bl{\ket{D_{i-1}}\ket{V_E(0,D_{i-1},R_i)}\ket{R_{i,[i,k-1]}}}\nonumber
\\&&\ \ \ket{V_E(S,R_1)}\ket{V_E(0,D_1,R_2)}\hdots\ket{V_E(0,D_{i-1},R_i)}\nonumber
\\&&\ \ \ \ket{V(0,D_i,R_{i+1})}\hdots\ket{V(0,D_{k-1},R_k)}\nonumber
\end{eqnarray*}
Now, this can be rearranged to get
\begin{eqnarray*}
\textcolor{blue}{\ket{S}}\sum_{\substack{(R_1,R_2,\hdots R_{i-1},\\R_{i,[i,k-1]},\\R_{i+1}\hdots R_k)\nonumber
\\\in\mathbb{F}_q^{m(k-1)-(i-1)b_i}}}\hspace{-0.5cm}
&&\textcolor{blue}{\ket{R_1}\ket{R_2}\hdots \ket{R_{i-1}}\ \ket{R_{i,[i,k-1]}}}\nonumber
\\[-1.4cm]&&\ \ \ket{V_E(S,R_1)}\ \ket{V_E(0,D_1,R_2)}\nonumber
\\&&\ \ \ \ \ \ \ \ \ \ \ \ \ \ \ \ \ \ \ \ \ \ \hdots\ket{V_E(0,D_{i-2},R_{i-1})}\nonumber
\\&&\ \ \ \ \ \ \ket{V(0,D_i,R_{i+1})}\hdots\ket{V(0,D_{k-1},R_k)}\nonumber
\\\sum_{\substack{R_{i,[i-1]}\\\in\mathbb{F}_q^{(i-1)b_i}}}&&\ket{V_E(0,D_{i-1},R_i)}\textcolor{blue}{\ket{V_E(0,D_{i-1},R_i)}}
\end{eqnarray*}
\begin{eqnarray*}
=\textcolor{blue}{\ket{S}}\!\!\!\sum_{\substack{(R_1,R_2,\hdots R_{i-1},\\R_{i,[i,k-1]},\\R_{i+1}\hdots R_k)\nonumber
\\\in\mathbb{F}_q^{m(k-1)-(i-1)b_i}}}\hspace{-0.5cm}
&&\textcolor{blue}{\ket{R_1}\ket{R_2}\hdots \ket{R_{i-1}}\ \ket{R_{i,[i,k-1]}}}\nonumber
\\[-1.4cm]&&\ \ \ket{V_E(S,R_1)}\ \ket{V_E(0,D_1,R_2)}\nonumber
\\&&\ \ \ \ \ \ \ \ \ \ \ \ \ \ \ \ \ \ \ \ \ \ \hdots\ket{V_E(0,D_{i-2},R_{i-1})}\nonumber
\\&&\ \ \ \ \ \ \ket{V(0,D_i,R_{i+1})}\hdots\ket{V(0,D_{k-1},R_k)}\nonumber
\\&&\ \ \ \ \ \ \sum_{T_i\in\mathbb{F}_q^{(i-1)b_i}}\ket{T_i}\textcolor{blue}{\ket{T_i}}
\end{eqnarray*}
because the state
\begin{equation*}
\sum_{\substack{R_{i,[i-1]}\\\in\mathbb{F}_q^{(i-1)\times b_i}}}\ket{V_E(0,D_{i-1},R_i)}\textcolor{blue}{\ket{V_E(0,D_{i-1},R_i)}}    
\end{equation*}
is a uniform superposition of states $\ket{T_i}\ket{T_i}$ over $T_i\in\mathbb{F}_q^{(i-1)\times b_i}$ independent of the value of $D_{i-1}$ and $R_{i,[i,k-1]}$.
\\\ \\
Since $D_{i-2}$ is formed from the $(i-2)$th rows of $R_1, R_2, \hdots, R_{i-2}$, the qudits can be rearranged to obtain,
\begin{eqnarray*}
\textcolor{blue}{\ket{S}}\!\!\sum_{\substack{\underline{r}\in\\\mathbb{F}_q^{m(k-1)}}}
&&\bl{\ket{R_{1,J_{i-2}}}\ket{R_{2,J_{i-2}}}\hdots \ket{R_{{i-2},J_{i-2}}}}\nonumber
\\[-0.7cm]&&\ \ \bl{\ket{D_{i-2}}\ket{R_{i-1}}\ket{R_{i,[i,k-1]}}}\nonumber
\\&&\ \ \ \ \ \ket{V_E(S,R_1)}\ket{V_E(0,D_1,R_2)}\hdots\ket{V_E(0,D_{i-2},R_{i-1})}\nonumber
\\&&\ \ \ \ \ \ \ \ket{V(0,D_i,R_{i+1})}\hdots\ket{V(0,D_{k-1},R_k)}\nonumber
\\&&\ \ \ \ \ \ \ \sum_{T_i\in\mathbb{F}_q^{(i-1)b_i}}\ket{T_i}\bl{\ket{T_i}}
\end{eqnarray*}
Applying $G_2$ on $\ket{D_{i-2}}\ket{R_{i-1}}$, we obtain,
\begin{eqnarray*}
\textcolor{blue}{\ket{S}}\!\!\sum_{\substack{\underline{r}\in\\\mathbb{F}_q^{m(k-1)}}}
\!\!\!\!&&\bl{\ket{R_{1,J_{i-2}}}\ket{R_{2,J_{i-2}}}\hdots \ket{R_{{i-2},J_{i-2}}}}\nonumber
\\[-0.7cm]&&\ \bl{\ket{D_{i-2}}\ket{V_E(0,D_{i-2},R_{i-1})}\ket{R_{i-1,[i,k-1]}}\ket{R_{i,[i,k-1]}}}\nonumber
\\&&\ \ \ket{V_E(S,R_1)}\ket{V_E(0,D_1,R_2)}\hdots\ket{V_E(0,D_{i-2},R_{i-1})}\nonumber
\\&&\ \ \ \ket{V(0,D_i,R_{i+1})}\hdots\ket{V(0,D_{k-1},R_k)}\nonumber
\\&&\ \ \ \ \ \ \ \sum_{T_i\in\mathbb{F}_q^{(i-1)b_i}}\ket{T_i}\bl{\ket{T_i}}
\end{eqnarray*}
Now, this can be rearranged to get
\begin{eqnarray*}
\textcolor{blue}{\ket{S}}\sum_{\substack{(R_1,R_2,\hdots R_{i-1},\\R_{i,[i,k-1]},\\R_{i+1}\hdots R_k)\nonumber
\\\in\mathbb{F}_q^{m(k-1)-(i-1)b_i}}}\hspace{-0.5cm}
&&\bl{\ket{R_1}\ket{R_2}\hdots \ket{R_{i-2}}\ \ket{R_{i-1,[i,k-1]}}\ket{R_{i,[i,k-1]}}}\nonumber
\\[-1.4cm]&&\ \ \ket{V_E(S,R_1)}\ \ket{V_E(0,D_1,R_2)}\nonumber
\\&&\ \ \ \ \ \ \ \ \ \ \ \ \ \ \ \ \ \ \ \ \ \ \hdots\ket{V_E(0,D_{i-3},R_{i-2})}\nonumber
\\&&\ \ \ \ \ \ \ket{V(0,D_i,R_{i+1})}\hdots\ket{V(0,D_{k-1},R_k)}\nonumber
\\\sum_{\substack{R_{i,[i-1]}\\\in\mathbb{F}_q^{(i-1)\times b_i}}}&&\ket{V_E(0,D_{i-2},R_{i-1})}\textcolor{blue}{\ket{V_E(0,D_{i-2},R_{i-1})}}\nonumber
\\&&\ \ \ \sum_{T_i\in\mathbb{F}_q^{(i-1)b_i}}\ket{T_i}\bl{\ket{T_i}}
\end{eqnarray*}
\begin{eqnarray*}
=\textcolor{blue}{\ket{S}}\!\!\!\!\!\!\sum_{\substack{(R_1,R_2,\hdots R_{i-1},\\R_{i,[i,k-1]},\\R_{i+1}\hdots R_k)\nonumber
\\\in\mathbb{F}_q^{m(k-1)-(i-1)b_i}}}\hspace{-0.5cm}
&&\!\!\!\!\textcolor{blue}{\ket{R_1}\ket{R_2}\hdots \ket{R_{i-2}}\ket{R_{i-1,[i,k-1]}}\ket{R_{i,[i,k-1]}}}\nonumber
\\[-1.4cm]&&\ \ \ket{V_E(S,R_1)}\ \ket{V_E(0,D_1,R_2)}\nonumber
\\&&\ \ \ \ \ \ \ \ \ \ \ \ \ \ \ \ \ \ \ \ \ \ \hdots\ket{V_E(0,D_{i-3},R_{i-2})}\nonumber
\\&&\ \ \ \ \ \ \ket{V(0,D_i,R_{i+1})}\hdots\ket{V(0,D_{k-1},R_k)}\nonumber
\\&&\hspace{-1.5cm}\sum_{T_{i-1}\in\mathbb{F}_q^{(i-1)\times b_{i-1}}}\ket{T_{i-1}}\bl{\ket{T_{i-1}}}\sum_{T_i\in\mathbb{F}_q^{(i-1)\times b_i}}\ket{T_i}\bl{\ket{T_i}}
\end{eqnarray*}
Performing similar operations with $\ket{R_j}$ for $1\leq j\leq i-2$, we obtain,
\begin{eqnarray*}
\textcolor{blue}{\ket{S}}\sum_{\substack{(R_{i+1}\hdots R_k)\\\in\mathbb{F}_q^{(m-a_i)(k-1)}\\(R_{1,[i,k-1]},\hdots R_{i,[i,k-1]})\\\in\mathbb{F}_q^{(k-i)a_i}}}\hspace{-0.5cm}
&&\textcolor{blue}{\ket{R_{1,[i,k-1]}}\ket{R_{2,[i,k-1]}}\hdots \ket{R_{i,[i,k-1]}}}\nonumber
\\[-1.4cm]&&\ \ \ket{V(0,D_i,R_{i+1})}\hdots\ket{V(0,D_{k-1},R_k)}\nonumber
\\\nonumber\\\nonumber
\\&&\hspace{-2.4cm}\sum_{\substack{T_1\in\\\mathbb{F}_q^{(i-1)\times b_1}}}\ket{T_1}\textcolor{blue}{\ket{T_1}}\sum_{\substack{T_2\in\\\mathbb{F}_q^{(i-1)\times b_2}}}\ket{T_2}\textcolor{blue}{\ket{T_2}}\hdots\sum_{\substack{T_i\in\\\mathbb{F}_q^{(i-1)\times b_i}}}\ket{T_i}\textcolor{blue}{\ket{T_i}}
\end{eqnarray*}
At this point the secret is completely disentangled with the rest of the qudits and the recovery is complete.
\end{proof}

\subsection{Secrecy}\label{ss:iii_c}
In the scheme given by Eq.~\eqref{eq:enc_qudits_univ_d}, the combiner can recover the secret by accessing $k$ parties (from case (ii) when $i=k$ in the proof of Lemma 1). So, by No-cloning theorem, the remaining $k-1$ parties in the scheme should have no information about the secret. Thus, this scheme satisfies the secrecy property.
Alternatively, we can invoke \cite[Theorem~5]{imai03} to show that  the secrecy requirement is met since all unauthorized sets are complements of authorized sets
in a threshold scheme.
With these results in place we have our central contribution. 

\begin{theorem}[Existence of universal communication efficient QTS]
There exists a QTS with the parameters $((k,2k-1,*))$ such that for all values of $1\leq i\leq k$ when any $d_i=2k-i$ parties are contacted by the combiner, the secret can be recovered from $\frac{m}{d_i-k+1}$ qudits received from each of the $d_i$ shares, where the secret contains $m$ qudits as in Eq.~\eqref{eq:ceqts-params}. 
\end{theorem}

In the standard $((k,n))$ QTS, the secret can be recovered when the combiner communicates with  $k$ parties. 
Here, if the secret is of size $m$ qudits, then the number of qudits communicated to  the combiner  is $km$ qudits.
The communication cost per secret qudit is $k$ qudits.

In the $((k,n,d))$ communication efficient QTS of \cite{senthoor19}, the secret can be recovered when the combiner contacts  $k$ parties and receiving  $km'$ qudits
where $m'=d-k+1$.
 This leads to  a cost of $k$ qudits per secret qudit.
However, when the combiner contacts $d$ parties, where $k<d\leq n$ is a fixed value, the secret can be recovered with a communication cost of $\frac{dm'}{d-k+1}$ qudits.
The cost per qudit is $\frac{d}{d-k+1}$ which is strictly less than $k$.

In the $((k,n,*))$ universal communication efficient QTS, the secret can be recovered by the combiner by accessing any $d_i$ parties, where the number of parties accessed given by $k\leq d_i\leq n$ is chosen by the combiner. For the chosen value of $d_i$, the secret can be recovered by downloading $\frac{d_i m}{d_i-k+1}$ qudits.
The per qudit communication cost is $\frac{d_i}{d_i-k+1}$ which is same as that of \cite{senthoor19}. 
However, we are able to achieve this for all $d_i$ using the same scheme and not fixing $d_i$ apriori.

\appendix
\section*{Example for $((3,5,*))$ communication efficient QTS}
\subsection{Parameters}
Take $k=3$. From Eq.~\eqref{eq:ceqts-params}, the parameters for the construction can be calculated as given below.
\begin{subequations}\label{eq:3-5-ceqts-params}
\begin{gather}
n=2k-1=5,q=11\\
d\in\{d_1,d_2,d_3\}\\
d_1=5, d_2=4, d_3=3\\
m_1=3, m_2=2, m_3=1\\
m=\text{lcm}\{m_1,m_2,m_3\}=6\\
a_1=2, a_2=3, a_3=6\\
b_1=2, b_2=1, b_3=3
\end{gather}
\end{subequations}

A secret of six qudits will be encoded into thirty qudits, giving six qudits for each party. Every qudit is of dimension 11.

\subsection{Encoding}
Encoding for this scheme can be given by the mapping
\begin{eqnarray}
\ket{s_1s_2s_3s_4s_5s_6}\ \mapsto\ \sum_{\underline{r}\in\F_{11}^{12}}
\begin{array}{l}
\ket{c_{1,1} c_{1,2} c_{1,3} c_{1,4} c_{1,5} c_{1,6}}
\\\ket{c_{2,1} c_{2,2} c_{2,3} c_{2,4} c_{2,5} c_{2,6}}
\\\ket{c_{3,1} c_{3,2} c_{3,3} c_{3,4} c_{3,5} c_{3,6}}
\\\ket{c_{4,1} c_{4,2} c_{4,3} c_{4,4} c_{4,5} c_{4,6}}
\\\ket{c_{5,1} c_{5,2} c_{5,3} c_{5,4} c_{5,5} c_{5,6}}
\end{array}
\label{eq:enc_qudits_3_5_ud_ce_qss_enc}
\end{eqnarray}
where $c_{ij}$ is the $(i,j)$th entry of the matrix  $C=VM$ and  $V$, $M$ are given below
\begin{eqnarray}
V=\left[
\begin{tabular}{ccccc}
9&3&4&6&1\\2&9&3&4&6\\8&2&9&3&4\\7&8&2&9&3\\5&7&8&2&9
\end{tabular}
\right],
\end{eqnarray}
\begin{eqnarray}
M=
\left[
\begin{tabular}{cc:c:ccc}
$s_1$ & $s_4$ & 0 & 0 & 0 & 0\\
$s_2$ & $s_5$ & $r_1$ & 0 & 0 & 0\\
$s_3$ & $s_6$ & $r_3$ & $r_2$ & $r_4$ & $r_6$\\\hdashline
$r_1$ & $r_3$ & $r_5$ & $r_7$ & $r_9$ & $r_{11}$\\
$r_2$ & $r_4$ & $r_6$ & $r_8$ & $r_{10}$ & $r_{12}$
\end{tabular}
\right].
\end{eqnarray}

Given these matrices $V$ and $M$, for this $((3,5,*))$ scheme,  
the encoded state in Eq.~\eqref{eq:enc_qudits_3_5_ud_ce_qss_enc} can be rewritten rearranging the qudits as follows. (With respect to Eq.~\eqref{eq:enc_qudits_3_5_ud_ce_qss_enc}, we have grouped the $i$th qudits of each party in the $i$th row below.)
\begin{eqnarray}
\!\!\!\sum_{\underline{r}\in\F_{11}^{12}}
\!\!\!\begin{array}{l}
\ket{c_{1,1} c_{2,1} c_{3,1} c_{4,1} c_{5,1}}
\\\ket{c_{1,2} c_{2,2} c_{3,2} c_{4,2} c_{5,2}}
\\\ket{c_{1,3} c_{2,3} c_{3,3} c_{4,3} c_{5,3}}
\\\ket{c_{1,4} c_{2,4} c_{3,4} c_{4,4} c_{5,4}}
\\\ket{c_{1,5} c_{2,5} c_{3,5} c_{4,5} c_{5,5}}
\\\ket{c_{1,6} c_{2,6} c_{3,6} c_{4,6} c_{5,6}}
\end{array}
\!\!\!=\sum_{\underline{r}\in\F_{11}^{12}}\!\!\!
\begin{array}{l}
\ket{V(s_1,s_2,s_3,r_1,r_2)}
\\\ket{V(s_4,s_5,s_6,r_3,r_4)}
\\\ket{V^{[2,5]}(r_1,r_3,r_5,r_6)}
\\\ket{V(0,0,r_2,r_7,r_8)}
\\\ket{V(0,0,r_4,r_9,r_{10})}
\\\ket{V(0,0,r_6,r_{11},r_{12})}.
\end{array}
\label{eq:enc_qudits_3_5_ud_ce_qss_enc_state}
\end{eqnarray}

For completeness, we give below the 
the encoded state in Eq.~\eqref{eq:enc_qudits_3_5_ud_ce_qss_enc}. 
\begin{eqnarray}
\sum_{\substack{(r_1,r_2,\hdots r_{12})\\\in\F_{11}^{12}}}
\begin{array}{l}
\ket{9s_1+3s_2+4s_3+6r_1+r_2}
\\\ \ket{9s_4+3s_5+4s_6+6r_3+r_4}
\\\ \ \ket{3r_1+4r_3+6r_5+r_6}
\\\ \ \ \ket{4r_2+6r_7+r_8}
\\\ \ \ \ \ket{4r_4+6r_9+r_{10}}
\\\ \ \ \ \ \ket{4r_6+6r_{11}+r_{12}}
\\\ket{2s_1+9s_2+3s_3+4r_1+6r_2}
\\\ \ket{2s_4+9s_5+3s_6+4r_3+6r_4}
\\\ \ \ket{9r_1+3r_3+4r_5+6r_6}
\\\ \ \ \ket{3r_2+4r_7+6r_8}
\\\ \ \ \ \ket{3r_4+4r_9+6r_{10}}
\\ \ \ \ \ \ \ket{3r_6+4r_{11}+6r_{12}}
\\\ket{8s_1+2s_2+9s_3+3r_1+4r_2}
\\\ \ket{8s_4+2s_5+9s_6+3r_3+4r_4}
\\\ \ \ket{2r_1+9r_3+3r_5+4r_6}
\\\ \ \ \ket{9r_2+3r_7+4r_8}
\\ \ \ \ \ \ket{9r_4+3r_9+4r_{10}}
\\\ \ \ \ \ \ket{9r_6+3r_{11}+4r_{12}}
\\\ket{7s_1+8s_2+2s_3+9r_1+3r_2}
\\\ \ket{7s_4+8s_5+2s_6+9r_3+3r_4}
\\\ \ \ket{8r_1+2r_3+9r_5+3r_6}
\\\ \ \ \ket{2r_2+9r_7+3r_8}
\\\ \ \ \ \ket{2r_4+9r_9+3r_{10}}
\\\ \ \ \ \ \ket{2r_6+9r_{11}+3r_{12}}
\\\ket{5s_1+7s_2+8s_3+2r_1+9r_2}
\\\ \ket{5s_4+7s_5+2s_6+8r_3+9r_4}
\\\ \ \ket{7r_1+2r_3+8r_5+9r_6}
\\\ \ \ \ket{2r_2+8r_7+9r_8}
\\\ \ \ \ \ket{2r_4+8r_9+9r_{10}}
\\ \ \ \ \ \ \ket{2r_6+8r_{11}+9r_{12}}
\end{array}\nonumber
\end{eqnarray}

\subsection{Secret Recovery}
For the encoding scheme given in Eq.~\eqref{eq:enc_qudits_3_5_ud_ce_qss_enc}, we can recover the secret from a subset of size $d\in\{3,4,5\}$. When $d=3$, each of the three accessed parties need to send all its six qudits. When $d=4$, each of the four accessed parties need to send only its first three qudits. When $d=5$, each of the five accessed parties need to send only its first two qudits.
We now show how to recover the secret for various sizes of the authorized set. 

\subsection*{Case 1 : $d=5$}
In this case, each of the five accessed parties sends only its first two qudits. Then the encoded state in Eq.~\eqref{eq:enc_qudits_3_5_ud_ce_qss_enc_state} can be rewritten as follows. (The basis states corresponding to the qudits accessed by the combiner are indicated in blue here.)
\begin{eqnarray*}
\sum_{\underline{r}\in\F_{11}^{12}}\!\!\!
\begin{array}{l}
\ket{\bl{c_{1,1} c_{2,1}\hdots c_{5,1}}}
\\\ket{\bl{c_{1,2} c_{2,2}\hdots c_{5,2}}}
\\\ket{c_{1,3} c_{2,3}\hdots c_{5,3}}
\\\ket{c_{1,4} c_{2,4}\hdots c_{5,4}}
\\\ket{c_{1,5} c_{2,5}\hdots c_{5,5}}
\\\ket{c_{1,6} c_{2,6}\hdots c_{5,6}}
\end{array}
=\sum_{\underline{r}\in\F_{11}^{12}}\!\!\!
\begin{array}{l}
\ket{\bl{V(s_1,s_2,s_3,r_1,r_2)}}
\\\ket{\bl{V(s_4,s_5,s_6,r_3,r_4)}}
\\\ket{V^{[2,5]}(r_1,r_3,r_5,r_6)}
\\\ket{V(0,0,r_2,r_7,r_8)}
\\\ket{V(0,0,r_4,r_9,r_{10})}
\\\ket{V(0,0,r_6,r_{11},r_{12})}
\end{array}
\end{eqnarray*}
Now, apply $V^{-1}$ to the first qudits of the five parties and then apply $V^{-1}$ to the second qudits of the five parties to obtain
\begin{eqnarray*}
\sum_{\underline{r}\in\F_{11}^{12}}
\begin{array}{l}
\ket{\bl{s_1,s_2,s_3,r_1,r_2}}
\\\ket{\bl{s_4,s_5,s_6,r_3,r_4}}
\\\ket{V^{[2,5]}(r_1,r_3,r_5,r_6)}
\\\ket{V(0,0,r_2,r_7,r_8)}
\\\ket{V(0,0,r_4,r_9,r_{10})}
\\\ket{V(0,0,r_6,r_{11},r_{12})}.
\end{array}
\end{eqnarray*}
On rearranging the qudits, we obtain the secret containing six qudits.
\begin{eqnarray*}
\ket{\bl{s_1 s_2 s_3 s_4 s_5 s_6}}
\sum_{\underline{r}\in\F_{11}^{12}}
\begin{array}{l}
\ket{\bl{r_1,r_2}}
\\\ket{\bl{r_3,r_4}}
\\\ket{V^{[2,5]}(r_1,r_3,r_5,r_6)}
\\\ket{V(0,0,r_2,r_7,r_8)}
\\\ket{V(0,0,r_4,r_9,r_{10})}
\\\ket{V(0,0,r_6,r_{11},r_{12})}
\end{array}
\end{eqnarray*}
Here, we have recovered any given basis state in the secret without any information leaking to the other qudits. Hence, the secret, which is an arbitrary superposition of the basis states, can also be recovered by the above operation. 

\subsection*{Case 2 : $d=4$}
Assume that the first four parties have been accessed by the combiner. Secret recovery for any other set of four parties will also happen in a similar way. In this case, each of the four accessed parties sends only its first three qudits. Then the encoded state can be rewritten as follows.
\begin{eqnarray*}
&&\sum_{\underline{r}\in\F_{11}^{12}}
\begin{array}{l}
\ket{\bl{c_{1,1} c_{2,1} c_{3,1} c_{4,1}}} \ket{c_{5,1}}
\\\ket{\bl{c_{1,2} c_{2,2} c_{3,2} c_{4,2}}} \ket{c_{5,2}}
\\\ket{\bl{c_{1,3} c_{2,3} c_{3,3} c_{4,3}}} \ket{c_{5,3}}
\\\ket{c_{1,4} c_{2,4} c_{3,4} c_{4,4} c_{5,4}}
\\\ket{c_{1,5} c_{2,5} c_{3,5} c_{4,5} c_{5,5}}
\\\ket{c_{1,6} c_{2,6} c_{3,6} c_{4,6} c_{5,6}}
\end{array}\nonumber
\\&&=\sum_{\underline{r}\in\F_{11}^{12}}
\begin{array}{l}
\ket{\bl{V_{[4]}(s_1,s_2,s_3,r_1,r_2)}}\ket{V_{\{5\}}(s_1,s_2,s_3,r_1,r_2)}
\\\ket{\bl{V_{[4]}(s_4,s_5,s_6,r_3,r_4)}}\ket{V_{\{5\}}(s_4,s_5,s_6,r_3,r_4)}
\\\ket{\bl{V_{[4]}^{[2,5]}(r_1,r_3,r_5,r_6)}}\ket{V_{\{5\}}^{[2,5]}(r_1,r_3,r_5,r_6)}
\\\ket{V(0,0,r_2,r_7,r_8)}
\\\ket{V(0,0,r_4,r_9,r_{10})}
\\\ket{V(0,0,r_6,r_{11},r_{12})}
\end{array}\nonumber
\end{eqnarray*}
The secret recovery happens in two parts. 
In the first part, we extract the basis state $\ket{s_1 s_2 s_3 s_4 s_5 s_6}$. 
In the second part, we disentangle the qudits containing the basis state from the remaining qudits.
\\\\1) To recover $\ket{r_1}$ and $\ket{r_3}$, apply ${V_{[4]}^{[2,5]}}^{-1}$ to $\ket{V_{[4]}^{[2,5]}(r_1,r_3,r_5,r_6)}$ to obtain
\begin{eqnarray*}
\sum_{\underline{r}\in\F_{11}^{12}}\!\!\!
\begin{array}{l}
\ket{\bl{V_{[4]}(s_1,s_2,s_3,r_1,r_2)}}\ket{V_{\{5\}}(s_1,s_2,s_3,r_1,r_2)}
\\\ket{\bl{V_{[4]}(s_4,s_5,s_6,r_3,r_4)}}\ket{V_{\{5\}}(s_4,s_5,s_6,r_3,r_4)}
\\\ket{\bl{r_1}}\ket{\bl{r_3}}\ket{\bl{r_5}}\ket{\bl{r_6}}
\ket{V_{\{5\}}^{[2,5]}(r_1,r_3,r_5,r_6)}
\\\ket{V(0,0,r_2,r_7,r_8)}
\\\ket{V(0,0,r_4,r_9,r_{10})}
\\\ket{V(0,0,r_6,r_{11},r_{12})}.
\end{array}
\end{eqnarray*}
Here, $\ket{V_{[4]}(s_1,s_2,s_3,r_1,r_2)}\ket{r_1}=\ket{W_1(s_1,s_2,s_3,r_1,r_2)}$ and $\ket{V_{[4]}(s_4,s_5,s_6,r_3,r_4)}\ket{r_3}=\ket{W_1(s_4,s_5,s_6,r_3,r_4)}$ where
\begin{equation*}
W_1=\left[
\begin{tabular}{c}
$V_{[4]}$\\\hline
0 0 0 1 0
\end{tabular}
\right].
\end{equation*}
\\\\2) To recover $\ket{s_1 s_2 s_3 s_4 s_5 s_6}$, apply $W_1^{-1}$ to the qudits $\ket{V_{[4]}(s_1,s_2,s_3,r_1,r_2)}\ket{r_1}$ and then apply $W_1^{-1}$ to the qudits $\ket{V_{[4]}(s_4,s_5,s_6,r_3,r_4)}\ket{r_3}$ to obtain
\begin{eqnarray*}
\sum_{\underline{r}\in\F_{11}^{12}}
\begin{array}{l}
\ket{\bl{s_1,s_2,s_3}}\ket{\bl{r_1}}\ket{V_{\{5\}}(s_1,s_2,s_3,r_1,r_2)}
\\\ket{\bl{s_4,s_5,s_6}}\ket{\bl{r_3}}\ket{V_{\{5\}}(s_4,s_5,s_6,r_3,r_4)}
\\\ket{\bl{r_2}}\ket{\bl{r_4}}\ket{\bl{r_5}}\ket{\bl{r_6}}\ket{V_{\{5\}}^{[2,5]}(r_1,r_3,r_5,r_6)}
\\\ket{V(0,r_2,r_7,r_8)}
\\\ket{V(0,r_4,r_9,r_{10})}
\\\ket{V(0,0,r_6,r_{11},r_{12})}.
\end{array}
\end{eqnarray*}
At this stage part of the $\ket{s_1,s_2,s_3}\ket{s_4,s_5,s_6}$ has been successfully extracted into a separate register. 
But  $\ket{s_1,s_2,s_3}$ is still entangled with $\ket{V_{\{5\}}(s_1,s_2,s_3,r_1,r_2)}\ket{r_1}\ket{r_2}$ and $\ket{s_4,s_5,s_6}$ is entangled with $\ket{V_{\{5\}}(s_4,s_5,s_6,r_3,r_4)}\ket{r_3}\ket{r_4}$. Further, $\ket{r_1}$ and $\ket{r_3}$ are entangled with $\ket{r_5}\ket{V_{\{5\}}^{[2,5]}(r_1,r_3,r_5,r_6)}$.
\\\ \\3) Consider the square matrix
\begin{equation*}
G_1=\left[
\begin{tabular}{c}
1 0 0 0\\
0 1 0 0\\\hline
$V_{\{5\}}^{[2,5]}$\\\hline
0 0 0 1
\end{tabular}
\right].
\end{equation*}
Now, apply $G_1$ to $\ket{r_1}\ket{r_3}\ket{r_5}\ket{r_6}$ to obtain,
\begin{eqnarray*}
\sum_{\underline{r}\in\F_{11}^{12}}
\begin{array}{l}
\ket{\bl{s_1,s_2,s_3}}\ket{\bl{r_1}}\ket{V_{\{5\}}(s_1,s_2,s_3,r_1,r_2)}
\\\ket{\bl{s_4,s_5,s_6}}\ket{\bl{r_3}}\ket{V_{\{5\}}(s_4,s_5,s_6,r_3,r_4)}
\\\ket{\bl{r_2}}\ket{\bl{r_4}}\ket{\bl{V_{\{5\}}^{[2,5]}(r_1,r_3,r_5,r_6)}}\ket{\bl{r_6}}\ket{V_{\{5\}}^{[2,5]}(r_1,r_3,r_5,r_6)}
\\\ket{V(0,0,r_2,r_7,r_8)}
\\\ket{V(0,0,r_4,r_9,r_{10})}
\\\ket{V(0,0,r_6,r_{11},r_{12})}
\end{array}
\end{eqnarray*}
Rearranging the qudits, we obtain
\begin{flalign}
&\sum_{\substack{(r_1,r_2\hdots r_4,\\r_6,r_7\hdots r_{12})\\\in\F_{11}^{11}}}
\begin{array}{l}
\ket{\bl{s_1,s_2,s_3}}\ket{\bl{r_1}}\ket{V_{\{5\}}(s_1,s_2,s_3,r_1,r_2)}
\\\ket{\bl{s_4,s_5,s_6}}\ket{\bl{r_3}}\ket{V_{\{5\}}(s_4,s_5,s_6,r_3,r_4)}
\\\ket{\bl{r_2}}\ket{\bl{r_4}}\ket{\bl{r_6}}
\\\ket{V(0,0,r_2,r_7,r_8)}
\\\ket{V(0,0,r_4,r_9,r_{10})}
\\\ket{V(0,0,r_6,r_{11},r_{12})}
\end{array}\label{eq:disentangle_d_4_1}
\\\!\!\!\!\!\!&\ \ \ \ \ \ \sum_{r_5\in\F_{11}}
\ket{V_{\{5\}}^{[2,5]}(r_1,r_3,r_5,r_6)}\ket{\bl{V_{\{5\}}^{[2,5]}(r_1,r_3,r_5,r_6)}}
\nonumber
\end{flalign}
For any given values of $r_1,r_3$ and $r_6$, the superposition of $\ket{V_{\{5\}}^{[2,5]}(r_1,r_3,r_5,r_6)}\ket{V_{\{5\}}^{[2,5]}(r_1,r_3,r_5,r_6)}$ over all values of $r_5\in\F_{11}$ will give the uniform superposition $\sum_{u\in\F_{11}}\ket{u}\ket{u}$, which is independent of $r_1,r_3$ and $r_6$. Hence, \eqref{eq:disentangle_d_4_1} can be simplified as
\begin{flalign*}
=&\sum_{\substack{(r_1,r_2\hdots r_4,\\r_6,r_7\hdots r_{12})\\\in\F_{11}^{11}}}
\begin{array}{l}
\ket{\bl{s_1,s_2,s_3}}\ket{\bl{r_1}}\ket{V_{\{5\}}(s_1,s_2,s_3,r_1,r_2)}
\\\ket{\bl{s_4,s_5,s_6}}\ket{\bl{r_3}}\ket{V_{\{5\}}(s_4,s_5,s_6,r_3,r_4)}
\\\ket{\bl{r_2}}\ket{\bl{r_4}}\ket{\bl{r_6}}
\\\ket{V(0,0,r_2,r_7,r_8)}
\\\ket{V(0,0,r_4,r_9,r_{10})}
\\\ket{V(0,0,r_6,r_{11},r_{12})}
\end{array}
\\&\ \ \ \ \ \ \sum_{f_5\in\F_{11}}
\ket{f_5}\ket{\bl{f_5}}
\end{flalign*}
4) Consider the square matrix
\begin{equation*}
G_2=\left[
\begin{tabular}{c}
1 0 0 0 0\\
0 1 0 0 0\\
0 0 1 0 0\\\hline
$V_{\{5\}}$\\\hline
0 0 0 0 1
\end{tabular}
\right].
\end{equation*}
Now, apply $G_2$ to $\ket{s_1,s_2,s_3}\ket{r_1}\ket{r_2}$ and then apply $G_2$ to $\ket{s_4,s_5,s_6}\ket{r_3}\ket{r_4}$ to obtain,
\begin{eqnarray*}
&\sum_{\substack{(r_1,r_2\hdots r_4,\\r_6,r_7\hdots r_{12})\\\in\F_{11}^{11}}}
\begin{array}{l}
\ket{\bl{s_1,s_2,s_3}}\ket{\bl{V_{\{5\}}(s_1,s_2,s_3,r_1,r_2)}}
\\\ \ \ \ \ \ \ \ \ \ \ \ \ \ \ \ \ \ \ \ \ \ \ \ \ \ket{V_{\{5\}}(s_1,s_2,s_3,r_1,r_2)}
\\\ket{\bl{s_4,s_5,s_6}}\ket{\bl{V_{\{5\}}(s_4,s_5,s_6,r_3,r_4)}}
\\\ \ \ \ \ \ \ \ \ \ \ \ \ \ \ \ \ \ \ \ \ \ \ \ \ \ket{V_{\{5\}}(s_4,s_5,s_6,r_3,r_4)}
\\\ket{\bl{r_2}}\ket{\bl{r_4}}\ket{\bl{r_6}}
\\\ket{V(0,0,r_2,r_7,r_8)}
\\\ket{V(0,0,r_4,r_9,r_{10})}
\\\ket{V(0,0,r_6,r_{11},r_{12})}
\end{array}\nonumber
\\&\sum_{f_5\in\F_{11}}
\ket{f_5}\ket{\bl{f_5}}
\end{eqnarray*}
Rearranging the qudits, we obtain
\begin{eqnarray}
&\hspace{-1.3cm}\ket{s_1 s_2 s_3 s_4 s_5 s_6} \sum_{\substack{(r_2,r_4,\\r_6,r_7\hdots r_{12})\\\in\F_{11}^{9}}}
\begin{array}{l}
\ket{\bl{r_2}}\ket{\bl{r_4}}\ket{\bl{r_6}}
\\\ket{V(0,0,r_2,r_7,r_8)}
\\\ket{V(0,0,r_4,r_9,r_{10})}
\\\ket{V(0,0,r_6,r_{11},r_{12})}
\end{array}\label{eq:disentangle_d_4_2}
\\&\hspace{2.7cm}\sum_{r_1\in\F_{11}}
\begin{array}{l}
\ket{V_{\{5\}}(s_1,s_2,s_3,r_1,r_2)}
\\\ket{\bl{V_{\{5\}}(s_1,s_2,s_3,r_1,r_2)}}
\end{array}\nonumber
\\&\hspace{3.3cm}\sum_{r_3\in\F_{11}}
\begin{array}{l}
\ket{V_{\{5\}}(s_4,s_5,s_6,r_3,r_4)}
\\\ket{\bl{V_{\{5\}}(s_4,s_5,s_6,r_3,r_4)}}
\end{array}\nonumber
\\&\hspace{1.8cm}\sum_{f_5\in\F_{11}}
\ket{f_5}\ket{\bl{f_5}}\nonumber
\end{eqnarray}
Similar to the argument below \eqref{eq:disentangle_d_4_1}, it can be proved that both the superposition of $\ket{V_{\{5\}}(s_1,s_2,s_3,r_1,r_2)}\ket{V_{\{5\}}(s_1,s_2,s_3,r_1,r_2)}$ over all values of $r_1$ and the superposition of $\ket{V_{\{5\}}(s_4,s_5,s_6,r_3,r_4)}\ket{V_{\{5\}}(s_4,s_5,s_6,r_3,r_4)}$ over all values of $r_3$ will give $\sum_{u\in\F_{11}}\ket{u}\ket{u}$. Hence, \eqref{eq:disentangle_d_4_2} can be simplified as
\begin{eqnarray*}
&\hspace{-1.8cm}\ket{s_1 s_2 s_3 s_4 s_5 s_6}\sum_{\substack{(r_1,r_2\hdots r_4,\\r_6,r_7\hdots r_{12})\\\in\F_{11}^{9}}}
\begin{array}{l}
\ket{\bl{r_2}}\ket{\bl{r_4}}\ket{\bl{r_6}}
\\\ket{V(0,0,r_2,r_7,r_8)}
\\\ket{V(0,0,r_4,r_9,r_{10})}
\\\ket{V(0,0,r_6,r_{11},r_{12})}
\end{array}\nonumber
\\&\hspace{1.7cm}\sum_{f_1\in\F_{11}}
\ket{f_1}\ket{\bl{f_1}}
\sum_{f_3\in\F_{11}}
\ket{f_3}\ket{\bl{f_3}}
\sum_{f_5\in\F_{11}}
\ket{f_5}\ket{\bl{f_5}}\nonumber
\end{eqnarray*}
Here, we have completely disentangled the basis states of the secret from the remaining qudits.
Hence, any arbitrary linear superposition of the basis states can be recovered by the above operations.

\subsection*{Case 3 : $d=3$}
Assume that the first three parties have been accessed by the combiner. Secret recovery for any other set of three parties will also happen in a similar way. In this case, each of the three accessed parties sends all its six qudits. Then the encoded state can be rewritten as follows.
\begin{eqnarray*}
&&\sum_{\underline{r}\in\F_{11}^{12}}
\begin{array}{l}
\ket{\bl{c_{1,1} c_{2,1} c_{3,1}}}\ket{c_{4,1} c_{5,1}}
\\\ket{\bl{c_{1,2} c_{2,2} c_{3,2}}}\ket{c_{4,2} c_{5,2}}
\\\ket{\bl{c_{1,3} c_{2,3} c_{3,3}}}\ket{c_{4,3} c_{5,3}}
\\\ket{\bl{c_{1,4} c_{2,4} c_{3,4}}}\ket{c_{4,4} c_{5,4}}
\\\ket{\bl{c_{1,5} c_{2,5} c_{3,5}}}\ket{c_{4,5} c_{5,5}}
\\\ket{\bl{c_{1,6} c_{2,6} c_{3,6}}}\ket{c_{4,6} c_{5,6}}
\end{array}\nonumber
\\&&=\sum_{\underline{r}\in\F_{11}^{12}}
\begin{array}{l}
\ket{\bl{V_{[3]}(s_1,s_2,s_3,r_1,r_2)}}\ket{V_{[4,5]}(s_1,s_2,s_3,r_1,r_2)}
\\\ket{\bl{V_{[3]}(s_4,s_5,s_6,r_3,r_4)}}\ket{V_{[4,5]}(s_4,s_5,s_6,r_3,r_4)}
\\\ket{\bl{V_{[3]}^{[2,5]}(r_1,r_3,r_5,r_6)}}\ket{V_{[4,5]}(0,r_1,r_3,r_5,r_6)}
\\\ket{\bl{V_{[3]}^{[3,5]}(r_2,r_7,r_8)}}\ket{V_{[4,5]}(0,0,r_2,r_7,r_8)}
\\\ket{\bl{V_{[3]}^{[3,5]}(r_4,r_9,r_{10})}}\ket{V_{[4,5]}(0,0,r_4,r_9,r_{10})}
\\\ket{\bl{V_{[3]}^{[3,5]}(r_6,r_{11},r_{12})}}\ket{V_{[4,5]}(0,0,r_6,r_{11},r_{12})}.
\end{array}
\end{eqnarray*}
Similar to $d=3$ case, the secret recovery happens in two parts. First, we will recover the basis state and then we will remove the entanglement with other qudits.
\\\\1) To recover $\ket{r_2}\ket{r_4}\ket{r_6}$, apply ${V_{[3]}^{[3,5]}}^{-1}$ to $\ket{V_{[3]}^{[3,5]}(r_2,r_7,r_8)}$, then to $\ket{V_{[3]}^{[3,5]}(r_4,r_9,r_{10})}$ and then to $\ket{V_{[3]}^{[3,5]}(r_6,r_{11},r_{12})}$ to obtain 
\begin{eqnarray*}
\sum_{\underline{r}\in\F_{11}^{12}}
\begin{array}{l}
\ket{\bl{V_{[3]}(s_1,s_2,s_3,r_1,r_2)}}\ket{V_{[4,5]}(s_1,s_2,s_3,r_1,r_2)}
\\\ket{\bl{V_{[3]}(s_4,s_5,s_6,r_3,r_4)}}\ket{V_{[4,5]}(s_4,s_5,s_6,r_3,r_4)}
\\\ket{\bl{V_{[3]}^{[2,5]}(r_1,r_3,r_5,r_6)}}\ket{V_{[4,5]}(0,r_1,r_3,r_5,r_6)}
\\\ket{\bl{r_2}}\ket{\bl{r_7}}\ket{\bl{r_8}}\ket{V_{[4,5]}(0,0,r_2,r_7,r_8)}
\\\ket{\bl{r_4}}\ket{\bl{r_9}}\ket{\bl{r_{10}}}\ket{V_{[4,5]}(0,0,r_4,r_9,r_{10})}
\\\ket{\bl{r_6}}\ket{\bl{r_{11}}}\ket{\bl{r_{12}}}\ket{V_{[4,5]}(0,0,r_6,r_{11},r_{12})}
\end{array}
\end{eqnarray*}
Here $\ket{V_{[3]}^{[2,5]}(r_1,r_3,r_5,r_6)}\ket{r_6}=\ket{W_2(r_1,r_3,r_5,r_6)}$ where
\begin{equation*}
W_2=\left[
\begin{tabular}{c}
$V_{[3]}^{[2,5]}$\\\hline
0 0 0 1
\end{tabular}
\right]
\end{equation*}
2) To recover $\ket{r_1}$ and $\ket{r_3}$, apply ${W_2}^{-1}$ to the qudits $\ket{V_{[3]}^{[2,5]}(r_1,r_3,r_5,r_6)}\ket{r_6}$ to obtain
\begin{eqnarray*}
\sum_{\underline{r}\in\F_{11}^{12}}
\begin{array}{l}
\ket{\bl{V_{[3]}(s_1,s_2,s_3,r_1,r_2)}}\ket{V_{[4,5]}(s_1,s_2,s_3,r_1,r_2)}
\\\ket{\bl{V_{[3]}(s_4,s_5,s_6,r_3,r_4)}}\ket{V_{[4,5]}(s_4,s_5,s_6,r_3,r_4)}
\\\ket{\bl{r_1}}\ket{\bl{r_3}}\ket{\bl{r_5}}\ket{V_{[4,5]}(0,r_1,r_3,r_5,r_6)}
\\\ket{\bl{r_2}}\ket{\bl{r_7}}\ket{\bl{r_8}}\ket{V_{[4,5]}(0,0,r_2,r_7,r_8)}
\\\ket{\bl{r_4}}\ket{\bl{r_9}}\ket{\bl{r_{10}}}\ket{V_{[4,5]}(0,0,r_4,r_9,r_{10})}
\\\ket{\bl{r_6}}\ket{\bl{r_{11}}}\ket{\bl{r_{12}}}\ket{V_{[4,5]}(0,0,r_6,r_{11},r_{12})}
\end{array}
\end{eqnarray*}
Here $\ket{V_{[3]}(s_1,s_2,s_3,r_1,r_2)}\ket{r_1}\ket{r_2}=\ket{W_2(s_1,s_2,s_3,r_1,r_2)}$ and $\ket{V_{[3]}(s_4,s_5,s_6,r_3,r_4)}\ket{r_3}\ket{r_4}=\ket{W_2(s_4,s_5,s_6,r_3,r_4)}$ where
\begin{equation*}
W_2=\left[
\begin{tabular}{c}
$V_{[3]}$\\\hline
0 0 0 1 0\\
0 0 0 0 1
\end{tabular}
\right]
\end{equation*}
3) To recover $\ket{s_1,s_2,s_3}$ and $\ket{s_4,s_5,s_6}$, apply $W_2^{-1}$ to $\ket{V_{[3]}(s_1,s_2,s_3,r_1,r_2)}\ket{r_1}\ket{r_2}$  and then to $\ket{V_{[3]}(s_4,s_5,s_6,r_3,r_4)}\ket{r_3}\ket{r_4}$ to obtain,
\begin{eqnarray*}
\sum_{\underline{r}\in\F_{11}^{12}}
\begin{array}{l}
\ket{\bl{s_1}}\ket{\bl{s_2}}\ket{\bl{s_3}}\ket{V_{[4,5]}(s_1,s_2,s_3,r_1,r_2)}
\\\ket{\bl{s_4}}\ket{\bl{s_5}}\ket{\bl{s_6}}\ket{V_{[4,5]}(s_4,s_5,s_6,r_3,r_4)}
\\\ket{\bl{r_1}}\ket{\bl{r_3}}\ket{\bl{r_5}}\ket{V_{[4,5]}(0,r_1,r_3,r_5,r_6)}
\\\ket{\bl{r_2}}\ket{\bl{r_7}}\ket{\bl{r_8}}\ket{V_{[4,5]}(0,0,r_2,r_7,r_8)}
\\\ket{\bl{r_4}}\ket{\bl{r_9}}\ket{\bl{r_{10}}}\ket{V_{[4,5]}(0,0,r_4,r_9,r_{10})}
\\\ket{\bl{r_6}}\ket{\bl{r_{11}}}\ket{\bl{r_{12}}}\ket{V_{[4,5]}(0,0,r_6,r_{11},r_{12})}
\end{array}
\end{eqnarray*}.
The basis state $\ket{s_1 s_2 s_3 s_4 s_5 s_6}$ has been successfully recovered. But still it is entangled with other qudits.
\\\ \\4) $\ket{r_2}$ is entangled with $\ket{r_7}\ket{r_8}\ket{V_{[4,5]}(0,0,r_2,r_7,r_8)}$, $\ket{r_4}$ is entangled with $\ket{r_9}\ket{r_{10}}\ket{V_{[4,5]}(0,0,r_4,r_9,r_{10})}$ and $\ket{r_6}$ is entangled with $\ket{r_{11}}\ket{r_{12}}\ket{V_{[4,5]}(0,0,r_6,r_{11},r_{12})}$. These entanglements will be removed in this step. Consider the matrix
\begin{equation*}
G_1=\left[
\begin{tabular}{c}
1 0 0\\\hline
$V_{[4,5]}^{[3,5]}$
\end{tabular}
\right]
\end{equation*}
Apply $G_1$ to $\ket{r_2}\ket{r_7}\ket{r_8}$, then to $\ket{r_4}\ket{r_9}\ket{r_{10}}$ and then to $\ket{r_6}\ket{r_{11}}\ket{r_{12}}$ to obtain,
\begin{eqnarray}
\sum_{\underline{r}\in\F_{11}^{12}}
\begin{array}{l}
\ket{\bl{s_1}}\ket{\bl{s_2}}\ket{\bl{s_3}}\ket{V_{[4,5]}(s_1,s_2,s_3,r_1,r_2)}
\\\ket{\bl{s_4}}\ket{\bl{s_5}}\ket{\bl{s_6}}\ket{V_{[4,5]}(s_4,s_5,s_6,r_3,r_4)}
\\\ket{\bl{r_1}}\ket{\bl{r_3}}\ket{\bl{r_5}}\ket{V_{[4,5]}(0,r_1,r_3,r_5,r_6)}
\\\ket{\bl{r_2}}\ket{\bl{V_{[4,5]}^{[3,5]}(r_2,r_7,r_8)}}\ket{V_{[4,5]}^{[3,5]}(r_2,r_7,r_8)}
\\\ket{\bl{r_4}}\ket{\bl{V_{[4,5]}^{[3,5]}(r_4,r_9,r_{10})}}\ket{V_{[4,5]}^{[3,5]}(r_4,r_9,r_{10})}
\\\ket{\bl{r_6}}\ket{\bl{V_{[4,5]}^{[3,5]}(r_6,r_{11},r_{12})}}\ket{V_{[4,5]}^{[3,5]}(r_6,r_{11},r_{12})}
\end{array}\label{eq:disentangle_d_3}
\end{eqnarray}
Using arguments similar to those under \eqref{eq:disentangle_d_4_1}, the state in \eqref{eq:disentangle_d_3} can be simplified as
\begin{eqnarray*}
\sum_{\substack{(r_1,r_2\hdots r_6\\F_{11},f_8\hdots f_{12})\\\in\F_{11}^{12}}}
\begin{array}{l}
\ket{\bl{s_1}}\ket{\bl{s_2}}\ket{\bl{s_3}}\ket{V_{[4,5]}(s_1,s_2,s_3,r_1,r_2)}
\\\ket{\bl{s_4}}\ket{\bl{s_5}}\ket{\bl{s_6}}\ket{V_{[4,5]}(s_4,s_5,s_6,r_3,r_4)}
\\\ket{\bl{r_1}}\ket{\bl{r_3}}\ket{\bl{r_5}}\ket{V_{[4,5]}(0,r_1,r_3,r_5,r_6)}
\\\ket{\bl{r_2}}\ket{\bl{f_7,f_8}}\ket{f_7,f_8}
\\\ket{\bl{r_4}}\ket{\bl{f_9,f_{10}}}\ket{f_9,f_{10}}
\\\ket{\bl{r_6}}\ket{\bl{f_{11},f_{12}}}\ket{f_{11},f_{12}}
\end{array}
\end{eqnarray*}
Rearranging the qudits, we obtain
\begin{eqnarray*}
\ket{\bl{s_1 s_2 s_3 s_4 s_5 s_6}}
\sum_{\substack{(r_1,r_2\hdots r_6\\F_{11},f_8\hdots f_{12})\\\in\F_{11}^{12}}}
\begin{array}{l}
\ket{\bl{r_1}}\ket{\bl{r_2}}\ket{V_{[4,5]}(s_1,s_2,s_3,r_1,r_2)}
\\\ket{\bl{r_3}}\ket{\bl{r_4}}\ket{V_{[4,5]}(s_4,s_5,s_6,r_3,r_4)}
\\\ket{\bl{r_5}}\ket{\bl{r_6}}\ket{V_{[4,5]}(0,r_1,r_3,r_5,r_6)}
\\\ket{\bl{f_7,f_8}}\ket{f_7,f_8}
\\\ket{\bl{f_9,f_{10}}}\ket{f_9,f_{10}}
\\\ket{\bl{f_{11},f_{12}}}\ket{f_{11},f_{12}}
\end{array}
\end{eqnarray*}
5) $\ket{r_1}$ and $\ket{r_3}$ are entangled with $\ket{r_5}\ket{r_6}$ $\ket{V_{[4,5]}(0,r_1,r_3,r_5,r_6)}$. This entanglement will be removed here. Consider the matrix
\begin{equation*}
G_2=\left[
\begin{tabular}{c}
1 0 0 0\\
0 1 0 0\\\hline
$V_{[4,5]}^{[2,5]}$
\end{tabular}
\right]
\end{equation*}
Apply $G_2$ to $\ket{r_1}\ket{r_3}\ket{r_5}\ket{r_6}$ to obtain
\begin{eqnarray*}
&&\ket{\bl{s_1 s_2 s_3 s_4 s_5 s_6}}
\sum_{\substack{(r_1,r_2\hdots r_6\\F_{11},f_8\hdots f_{12})\\\in\F_{11}^{12}}}
\!\!\!\begin{array}{l}
\ket{\bl{r_1}}\ket{\bl{r_2}}\ket{V_{[4,5]}(s_1,s_2,s_3,r_1,r_2)}
\\\ket{\bl{r_3}}\ket{\bl{r_4}}\ket{V_{[4,5]}(s_4,s_5,s_6,r_3,r_4)}
\\\ket{\bl{V_{[4,5]}^{[2,5]}(r_1,r_3,r_5,r_6)}}\\\ \ \ \ \ \ \ket{V_{[4,5]}^{[2,5]}(r_1,r_3,r_5,r_6)}
\\\ket{\bl{f_7,f_8}}\ket{f_7,f_8}
\\\ket{\bl{f_9,f_{10}}}\ket{f_9,f_{10}}
\\\ket{\bl{f_{11},f_{12}}}\ket{f_{11},f_{12}}
\end{array}\nonumber
\\&&=\ket{\bl{s_1 s_2 s_3 s_4 s_5 s_6}}
\!\!\!\sum_{\substack{(r_1,r_2\hdots r_4\\f_5,f_6\hdots f_{12})\\\in\F_{11}^{12}}}
\!\!\!\begin{array}{l}
\ket{\bl{r_1}}\ket{\bl{r_2}}\ket{V_{[4,5]}(s_1,s_2,s_3,r_1,r_2)}
\\\ket{\bl{r_3}}\ket{\bl{r_4}}\ket{V_{[4,5]}(s_4,s_5,s_6,r_3,r_4)}
\\\ket{\bl{f_5,f_6}}\ket{f_5,f_6}
\\\ket{\bl{f_7,f_8}}\ket{f_7,f_8}
\\\ket{\bl{f_9,f_{10}}}\ket{f_9,f_{10}}
\\\ket{\bl{f_{11},f_{12}}}\ket{f_{11},f_{12}}
\end{array}
\end{eqnarray*}
6) $\ket{s_1,s_2,s_3}$ is entangled with $\ket{V_{[4,5]}(s_1,s_2,s_3,r_1,r_2)}\ket{r_1}$ $\ket{r_2}$ and $\ket{s_4,s_5,s_6}$ is entangled with $\ket{V_{[4,5]}(s_4,s_5,s_6,r_3,r_4)}$ $\ket{r_3}\ket{r_4}$. Consider the matrix
\begin{equation*}
G_3=\left[
\begin{tabular}{c}
1 0 0 0 0\\
0 1 0 0 0\\
0 0 1 0 0\\\hline
$V_{[4,5]}$
\end{tabular}
\right]
\end{equation*}
Now, apply $G_3$ to $\ket{s_1 s_2 s_3}\ket{r_1}\ket{r_2}$ and then to $\ket{s_4 s_5 s_6}\ket{r_3}\ket{r_4}$ to obtain
\begin{eqnarray}
\ket{\bl{s_1 s_2 s_3 s_4 s_5 s_6}}
\sum_{\substack{(r_1,r_2\hdots r_4\\f_5,f_6\hdots f_{12})\\\in\F_{11}^{12}}}
\begin{array}{l}
\ket{\bl{V_{[4,5]}(s_1,s_2,s_3,r_1,r_2)}}\\\ \ \ \ \ \ \ \ \ \ \ket{V_{[4,5]}(s_1,s_2,s_3,r_1,r_2)}
\\\ket{\bl{V_{[4,5]}(s_4,s_5,s_6,r_3,r_4)}}\\\ \ \ \ \ \ \ \ \ \ \ket{V_{[4,5]}(s_4,s_5,s_6,r_3,r_4)}
\\\ket{\bl{f_5,f_6}}\ket{f_5,f_6}
\\\ket{\bl{f_7,f_8}}\ket{f_7,f_8}
\\\ket{\bl{f_9,f_{10}}}\ket{f_9,f_{10}}
\\\ket{\bl{f_{11},f_{12}}}\ket{f_{11},f_{12}}
\end{array}\nonumber
\\=\ket{\bl{s_1 s_2 s_3 s_4 s_5 s_6}}
\sum_{\substack{(f_1,f_2\hdots f_{12})\\\in\F_{11}^{12}}}
\begin{array}{l}
\ket{\bl{f_1,f_2}}\ket{f_1,f_2}
\\\ket{\bl{f_3,f_4}}\ket{f_3,f_4}
\\\ket{\bl{f_5,f_6}}\ket{f_5,f_6}
\\\ket{\bl{f_7,f_8}}\ket{f_7,f_8}
\\\ket{\bl{f_9,f_{10}}}\ket{f_9,f_{10}}
\\\ket{\bl{f_{11},f_{12}}}\ket{f_{11},f_{12}}
\end{array}\nonumber
\end{eqnarray}
Here, we have recovered any given basis state in the secret without any entanglement to the other qudits. Hence, the secret can be recovered by the above operations.

\bibliographystyle{plain}

\end{document}